\documentclass[letterpaper, 10 pt, conference]{ieeeconf}  

\IEEEoverridecommandlockouts                              

\RequirePackage{etex}
\usepackage{hyperref}
\usepackage{verbatim} 
\usepackage{amsmath} 
\usepackage{amssymb}  
\usepackage{ntheorem}
\usepackage{algorithm}
\usepackage{mathtools}
\usepackage{siunitx}
\usepackage{tabularx}
\usepackage[usenames,dvipsnames]{xcolor}
\usepackage{psfrag,amsbsy,graphics,float}
\usepackage[dvips]{graphicx}
\usepackage{pgfplots}
\usepackage{pgfplotstable}
\usepackage{multirow}
\usepgfplotslibrary{fillbetween}
\usepackage{booktabs}
\pgfplotsset{compat = newest}
\usetikzlibrary{arrows,positioning,shapes,intersections,patterns,calc,fit} 
\usepackage{cite}

\newtheorem{defn}{Definition}
\newtheorem{rem}{Remark}

\newtheorem{lem}{Lemma}
\newtheorem{cor}{Corollary}

\newtheorem{assum}{Assumption}

\newcommand\tran{\mkern-2mu\raise1.25ex\hbox{$\scriptscriptstyle\top\hspace{0.5mm}$}\mkern-3.5mu}
\newcommand{\R}{\mathbb{R}}

\newcommand{\N}{\mathbb{N}}

\newcommand{\M}{\mathcal{M}}
\newcommand{\D}{\mathcal{D}}

\newcommand{\K}{\mathcal{K}}

\newcommand{\bm}[1]{{\boldsymbol{#1}}}
\newcommand{\Verts}[1]{{\left\Vert #1 \right\Vert}}
\DeclareMathOperator{\diag}{diag}

\DeclareMathOperator*{\argmin}{argmin}

\newcommand{\bma}{\bm a}
\newcommand{\x}{\bm x}
\newcommand{\m}{\bm m}

\newcommand{\q}{\bm q}
\newcommand{\dq}{\dot{\bm q}}
\newcommand{\e}{\bm e}
\newcommand{\de}{\dot{\bm e}}
\newcommand{\ddq}{\ddot{\bm q}}
\newcommand{\f}{\bm{f}}
\newcommand{\g}{\bm{g}}
\renewcommand{\u}{\bm{u}}
\newcommand{\y}{\bm{y}}

\usepackage[noabbrev]{cleveref} 
\crefname{rem}{Remark}{Remarks}
\crefname{exam}{Example}{Examples}
\crefname{assum}{Assumption}{Assumptions}
\crefname{prop}{Proposition}{Propositions}
\crefname{cor}{Corollary}{Corollaries}
\crefname{lem}{Lemma}{Lemmas}
\crefname{thm}{Theorem}{Theorems}
\crefname{defn}{Definition}{Definitions}
\crefname{figure}{Fig.}{Fig.}
\Crefname{figure}{Figure}{Figures}
\crefname{equation}{}{}

\usepackage{autonum}

\title{\LARGE \bf
Closed-loop Model Selection for Kernel-based Models\\ using Bayesian Optimization
}

\author{Thomas Beckers$^{1}$, Somil Bansal$^{2}$, Claire J. Tomlin$^{2}$ and Sandra Hirche$^{1}$
\thanks{$^{1}$ are with the Chair of Information-oriented Control (ITR), Department of Electrical and Computer Engineering,
Technical University of Munich, 80333 Munich, Germany, {\tt\footnotesize \{t.beckers, hirche\}@tum.de}\newline
$^{2}$ are with the Department of Electrical Engineering and Computer Sciences,
UC, Berkeley, USA, {\tt\footnotesize \{somil, tomlin\}@eecs.berkeley.edu}}
}

\begin{document}

\maketitle
\thispagestyle{empty}
\pagestyle{empty}

\begin{abstract}
Kernel-based nonparametric models have become very attractive for model-based control approaches for nonlinear systems. However, the selection of the kernel and its hyperparameters strongly influences the quality of the learned model. Classically, these hyperparameters are optimized to minimize the prediction error of the model but this process totally neglects its later usage in the control loop. 
In this work, we present a framework to optimize the kernel and hyperparameters of a kernel-based model directly with respect to the closed-loop performance of the model.
Our framework uses Bayesian optimization to iteratively refine the kernel-based model using the observed performance on the actual system until a desired performance is achieved.
We demonstrate the proposed approach in a simulation and on a 3-DoF robotic arm.
\end{abstract}

\section{Introduction}
Given a dynamic model, control mechanisms such as model predictive control and feedback linearization can be used to effectively control nonlinear systems. 
However, when an accurate mathematical model of the system is not available, machine learning offers powerful tools for the modeling of dynamical systems. 
A special class of models that has obtained a lot of attention recently is kernel-based models, such as Support Vector Machines (SVM) and Gaussian Processes (GP).
In contrast to parametric models, kernel-based models require only minimal prior knowledge about the system dynamics, and have been sucessfully used to model complex, nonlinear systems~\cite{bishop2006pattern}.
Using the kernel-based approach for modeling a system requires the selection of an appropriate kernel function and a set of hyperparameters for that function.
Typically, these selections are data-based, e.g. through minimizing a loss function that is often a trade-off between the prediction error and the complexity of the model.
However, the full complex and accurate dynamics model might not even be required depending on the task. 
Moreover, this procedure neglects the fact that the learned model is used for the control of the actual system, which can result in reduced controller performance~\cite{abbeel2006using,geversaa2005identification,hjalmarsson1996model}.

In this work, we propose a Bayesian Optimization (BO)-based active learning framework to optimize the kernel and its hyperparameters directly with respect to the performance of the closed-loop rather than the prediction error, see~\cref{fig:sys_bsb}. 
This optimization is performed in a sequential fashion where at each step of the optimization, BO takes into account all the past data points and proposes the most promising kernel and hyperparameters for the next evaluation.
The outcome is used to define a kernel-based model that is utilized by a given controller. The obtained model-based controller is then applied to the actual system in a closed-loop fashion to evaluate its performance. This information is then used by BO to optimize the next evaluation. Consequently, multiple evaluations on the actual system must be performed, which is often feasible such as for systems with repetitive trajectories. BO thus does not aim to obtain the most accurate dynamics model of the system, but rather to optimize the performance of the closed-loop system.
\begin{figure}[t]
\begin{center}
\vspace{0.15cm}
	\begin{tikzpicture}[auto,>=latex]
\tikzstyle{block} = [draw, fill=white, rectangle, minimum height=0.8cm, minimum width=2cm, align=center,inner sep=0.5mm]
\tikzstyle{sum} = [draw, fill=white, circle, node distance=1cm]
\tikzstyle{input} = [coordinate]
\tikzstyle{output} = [coordinate]
\tikzstyle{via} = [coordinate]

    \node [block] (controller) {Controller};
    \node [via,  above= 0.3 cm of controller] (viabl) {};
    \node [sum, left= 0.5cm of controller,label={[xshift=-3mm,yshift=-2mm]center:-}] (sumbl) {};
    \node [via,  right= 0.4cm of controller] (viabl1) {};
    \node [block, right= 0.8cm of controller] (system) {System};
    \node [block, below= 0.3 cm of controller] (GPR) {Kernel-based\\model};
    \node [block, left= 0.5cm of GPR] (BO) {Bayesian\\optimization};
    \node [block, below= 0.3 cm of BO] (CF) {Cost\\function};
    
    \node [input, left= 1.2cm of sumbl] (input) {};
    \node [output, right= 0.3cm of system] (output) {};
    
    \draw [->] (input) -- node[pos=0.2] {Reference} (sumbl);
    \draw [->] (sumbl) -- (controller);
    \draw [->] (controller) -- (system);
    \draw [-] (system) -- (output);
    \draw [-] (output) |- (viabl);
    \draw [->] (viabl) -| (sumbl);
    \draw [->] (viabl1) |- ([yshift=0.2cm]GPR.east);
    \draw [->] (output) |- ([yshift=-0.2cm]GPR.east);
    \draw [->] (viabl1) |- ([yshift=0.2cm]CF.east);
    \draw [->] (output) |- ([yshift=-0.2cm]CF.east);
    \draw [->] (CF) -- (BO);
    \draw [->] (BO) -- (GPR);
    \draw [->] (GPR) -- (controller);
    \draw [->] (input) -- (sumbl);
    
%
%
		
	\node[draw,dashed,inner xsep=0.2cm,inner ysep=0.2cm,fit=(BO) (CF)] {};

\end{tikzpicture}
	\vspace{-0.2cm}\caption{Closed-loop model selection for kernel-based models. BO is used to optimize the kernel and its hyperparameters directly based on the evaluation of a cost function.}
	\vspace{-0.8cm}
	\label{fig:sys_bsb}
\end{center}
\end{figure}
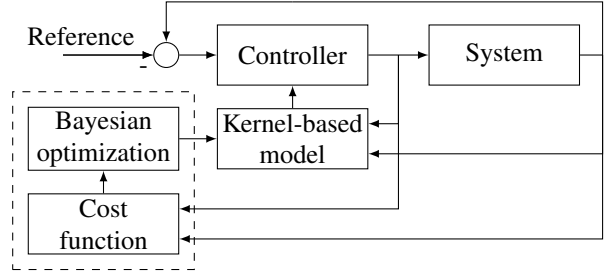

Typically, system identification approaches aim to obtain an open-loop dynamics model of the system by minimizing the state prediction error.
This problem has been well studied in literature for both linear systems, e.g.~\cite{ljung1998system}, as well as for nonlinear systems using the function approximators such as GP~\cite{chowdhary2012model, umlauft:cdc2017, beckers:ifacwc2017} and neural networks (NN)~\cite{narendra1990identification, bansal2016learning}. 
However, a model obtained using this open-loop procedure can result in a reduced controller performance on the actual system~\cite{hjalmarsson1996model}. To overcome these challenges, adaptive control mechanisms and iterative learning control have been studied where the system dynamics or control parameters are optimized based on the performance on the actual system, e.g. ~\cite{aastrom2013adaptive, clarke1985generalized,bristow2006survey}. However, these approaches are mostly limited to linear systems and controllers or assume at least a parametric system model. Recently, learning-based controller tuning mechanisms have also been proposed~\cite{lewis2013reinforcement, Calandra2015a}, but such methods might be highly data-inefficient for general nonlinear systems as they typically completely disregard the underlying dynamics~\cite{recht2018tour}. 

To overcome the challenges of open-loop system identification, closed-loop system identification methods have been studied that lead to more robust control performance on the actual system~\cite{abbeel2006using,geversaa2005identification,hjalmarsson1996model}.
A similar approach is presented in \cite{bansal2017goal}, wherein the authors also propose a goal-driven dynamics learning approach via BO.
However, the authors aim to identify a linear dynamics model from \textit{scratch} which might be a) unnecessary, as often an approximate dynamics model of the system is available and b) insufficient for general nonlinear systems. Moreover, stability of the closed-loop system where the controller is based on the linear dynamics model cannot be guaranteed, whereas our approach explicitly allows to preserve the convergence properties of the initial closed-loop system.
To summarize, our key contributions are: a) we present a BO-based framework to optimize the kernel function and its hyperparameters of a kernel-based model to maximize the resultant control performance on the actual system; b) through numerical examples and an experiment on a real 3-DoF robot, we demonstrate the advantages of the proposed approach over classical model selection methods. 

\noindent \textbf{Notation:} Vectors $\bm{a}$ are denoted with bold characters. Matrices $A$ are described with capital letters. The term~$A_{i,:}$ denotes the i-th row of the matrix~$A$. The expression~$\mathcal{N}(\mu,\Sigma)$ describes a normal distribution with mean~$\mu$ and covariance~$\Sigma$. The set $\R_{>0}$ denotes the set of positive real numbers.
\section{Problem Setting}\label{sec:ps}
Consider a discrete-time, potentially nonlinear system
\begin{align}
\begin{split}
\x_{k+1}&=\f(\x_k,\u_k),\quad k=\{0,\ldots,n-1\},n\in\N\\
\y_k&=\g(\x_k,\u_k)\label{for:sys}
\end{split}
\end{align}
in which $\f$, $\g$ are unknown functions of the state $\x_k\in\R^{n_x}$ and input $\u_k\in\R^{n_u}$. For the following, we assume that the state mapping $\f\colon\R^{n_x}\times\R^{n_u}\to\R^{n_x}$ and the output mapping $\g\colon\R^{n_x}\times\R^{n_u}\to\R^{n_y}$ are such that there exist a unique state and output trajectory for all $\u_k\in\R^{n_u}$ and $\x_0,~k\geq 0$. 
We assume that a control law $\bm{h}\colon\R^{n_y}\times\R^{n_m}\to\R^{n_u}$
\begin{align}
\u_k=\bm{h}(\y_k-\bm{r}_k,\bm{m}_k)\label{for:ctrllaw}
\end{align}
is given for the system~\cref{for:sys}. The reference $\bm{r}_k\in\R^{n_y}$ is assumed to be zero but the framework is also applicable for a varying signal. In addition to the reference, the control law also depends on the output $\bm{m}_k\in\R^{n_m}$ of a kernel-based model, a regression technique that uses a kernel to perform the regression in a higher-dimensional feature space. The output of a kernel-based model, $\bm{m}_k$, depends on the kernel function $\mathfrak{K}$, its hyperparameters $\bm{\varphi}\in\R^{n_\varphi}$ and system input and output, i.e. $\bm{m}_k=\M(\u_{0:k-1},\y_{0:k},\mathfrak{K},\bm{\varphi})$, where the function $\M$ depends on the class of the kernel-based model, such as GP or SVM, used for the prediction.
\begin{rem}\label{rem:km}
    For example, the output $\bm{m}_k$ can be the prediction of the next state or output of the system based on the current state and input, using the mean and probably the variance of a GP model. This information can then be used by the controller to compute an appropriate system input $\u_k$.  
\end{rem}
The control law $\bm{h}$ might be an output tracking controller designed based on the predicted model output. For possible control laws for different classes of systems, we refer to~\cite{chowdhary2012model, umlauft:cdc2017, beckers:ifacwc2017,suykens2001optimal,berkenkamp2016safe}. The goal of this work is to optimize the model kernel and its hyperparameters such that the corresponding model output $\m_k$ minimizes the following cost functional
\begin{align}
C(\y_{0:k},\u_{0:k})=\sum_{k=0}^{n-1}c(\y_k,\u_k),
\label{for:cfcn}
\end{align}
where $c(\y_k,\u_k)\colon\R^{n_y}\!\times\R^{n_u}\!\to\!\R$ represents the cost incurred for the control input $\u_k$ and the system output $\y_k$.
The cost function here might represent the requirements concerning the closed-loop, e.g. an accurate tracking behavior or a minimized power consumption. 
Note that the cost functional in \cref{for:cfcn} implicitly depends on the kernel-based model $\M$ through $\u_k$, see \cref{for:ctrllaw}.
The optimization of~\cref{for:cfcn} is challenging since the system dynamics in~\cref{for:sys} are unknown and the kernel-based model output $\bm{m}_k$ indirectly influences the cost. To overcome this challenge, we use BO to optimize the kernel and the hyperparameters based on the direct evaluation of the control law in~\cref{for:ctrllaw} on the system~\cref{for:sys} to find those that minimize the cost functional in~\cref{for:cfcn}.
\section{Preliminaries}\label{sec:pre}
\subsection{Kernel-based models}
The prediction of parametric models is based on a parameter vector $\bm{w}\in\R^{n_a}$ which is typically learned using a set of training data points. In contrast, nonparametric models typically maintain a subset of the training data points in memory in order to make predictions for new data points. Many linear models can be transformed into a dual representation where the prediction is based on a linear combination of kernel functions. The idea is to transform the data points of a model to an often high-dimensional feature space where a linear regression can be applied to predict the model output. For a nonlinear feature map $\bm{\phi}\colon\R^{n_a}\!\to\!\R^{n_\phi}$ with $n_\phi\!\in\!\N\!\cup\{\infty\}$, the kernel function is given by the inner product $\mathfrak{K}(\bma,\bma^\prime)\!=\!\langle \bm{\phi}(\bma),\bm{\phi}(\bma^\prime)\rangle,\forall\bma,\bma^\prime\!\in\!\R^{n_a}$.\\
Thus, the kernel implicitly encodes the way the data points are transformed into a higher dimensional space. The formulation as inner product in a feature space allows to extend many standard regression methods. A drawback of kernel-based models is that the selection of the kernel and its hyperparameters heavily influences the interpretation of the data and thus, the quality of the model. Commonly, the kernel and hyperparameters are determined based on the optimization of a loss function such as cross-validation or the likelihood function. In our work, the kernel and its hyperparameters are optimized with respect to performance of the closed-loop system. 

\subsection{Gaussian process}
Extending the concept of kernel functions to probabilistic models leads to the framework of Gaussian process regression (GPR). In particular, GPR is a supervised learning technique which combines several advantages. As probabilistic kernel techniques, GPs provides not only a mean function but also a measure for the uncertainty of the regression. In this work, we use GPR in BO to model the unknown \textit{closed-loop} objective function, as well as for the kernel-based dynamics model $\M$ in the experiment. The GPR can be derived using a standard linear regression model 
\begin{align}
q(\bma)=\bma^\top\bm{w},\quad b=q(\bma)+\epsilon
\end{align}
where $\bma\in\R^{n_a}$ is the input vector, $\bm{w}$ the vector of weights and $q\colon\!\R^{n_a}\!\!\to\!\R$ the function value. The observed value $b\in\!\R$ is corrupted by Gaussian noise $\epsilon\sim\mathcal{N}(0,\sigma_n^2)$. Using the feature map $\bm{\phi}(\bma)$ instead of $\bma$, leads to $f(\bma)=\bm{\phi}(\bma)^\top\bm{w}$ with $f\colon\R^{n_a}\to\R$. The analysis of this model is analogous to the standard linear regression, i.e. we put a prior on the weights such that $\bm{w}\sim\mathcal{N}(\bm{0},\Sigma_p)$ with $\Sigma_p\in\R^{n_\phi\times n_\phi}$. The mean function is usually defined to be zero, see~\cite{rasmussen2006gaussian}. Based on $m$ collected training data points $A=[\bma_1,\ldots,\bma_m]$ and $B=[b_1,\ldots,b_m]^\top$, the prediction $q_*\in\R$ for a new test point $\bma_*\in\R^{n_a}$ can be computed using the Bayes' rule. In particular, it is given by 
\begin{align}
q_*\vert\bma_*,A,B\!\sim\!\mathcal{N}(\bm{k}_*^\top K_{**}^{-1}B, k_{**}\!-\!\bm{k}_*^\top K_{**}^{-1}\bm{k}_*),\label{for:gp2}
\end{align}
where $\mathfrak{K}(\bma,\bma^\prime)=\phi(\bma)^\top\Sigma_p\phi(\bma^\prime)$, $k_{**}=\mathfrak{K}(\bma_*,\bma_*)$ and $\bm{k}_*=[\mathfrak{K}(\bma_*,A_{1,:}),\ldots,\mathfrak{K}(\bma_*,A_{m,:})]^\top$. The covariance matrix $K_{**}=(K+\sigma_n^2 I)$ is defined by $K_{i,j}=\mathfrak{K}(\bm{a}_i,\bm{a}_j)$. Thus, based on the training data $A,B$, the estimation of the function value $q_*$ follows a normal distribution where the mean and the variance depend on the test input $a_*$. Following~\cref{rem:km}, the mean and variance can be used for state estimation in the control law~\cref{for:ctrllaw}. The choice of the kernel and hyperparameters $\bm{\varphi}\in\R^{n_\varphi}$ can be seen as degrees of freedom of the regression. A popular kernel choice in GPR is the squared exponential kernel, see~\cite{rasmussen2006gaussian}. 
One possibility for estimating the hyperparameters~$\bm{\varphi}$ is by means of the likelihood function, thus by maximizing the probability of
\begin{align}
\bm{\varphi}^* &= \frac{1}{2}\left(B^\top K_{**}^{-1}B+\log\vert K_{**} \vert+m\log 2\pi\right)\label{for:likopt}
\end{align}
which results in an automatic trade-off between the data-fit $B^\top K_{**}^{-1}B$ and model complexity $\log\vert K_{**} \vert$, see~\cite{rasmussen2006gaussian}. 

\subsection{Bayesian Optimization (BO)}
Bayesian Optimization is an approach to minimize an unknown objective function based on (only a few) evaluated samples. We use BO to optimize the cost function~\cref{for:cfcn} based on the kernel-based model as this is in general a non-convex optimization problem with unknown objective function (because the system dynamics are unknown), and probably multiple local extrema. BO is well-suited for this optimization as the task evaluations can be expensive and noisy~\cite{shahriari2016taking}. Futhermore, BO is a gradient-free optimization method which only requires that the objective function can be evaluated for any given input. Since the objective function is unknown, the Bayesian strategy is to treat it as a random function with a prior, often as Gaussian process. Note that this GP here is used for the closed-loop cost functional approximation in BO and is not related to the kernel-based model for the controller~\cref{for:ctrllaw} as stated in~\cref{rem:km}. The prior captures the beliefs about the behaviour of the function, e.g. continuity or boundedness. After gathering the cost~\cref{for:cfcn} of the task evaluation, the prior is updated to form the posterior distribution over the objective function. The posterior distribution is used to construct an acquisition function that determines the most promising kernel/hyperparameters for the next evaluation to minimize the cost. Different acquisition functions are used in literature to trade off between exploration of unseen kernel/hyperparameters and exploitation of promising combinations during the optimization process. Common acquisition functions are expected improvement, entropy search, and upper confidence bound~\cite{mockus2012bayesian}. To escape a local objective function minimum, the authors of~\cite{bull2011convergence} propose a method to modify the acquisition function when they seem to over-exploit an area, namely expected-improvement-plus. That results in a more comprehensive and also partially random exploration of the area and, thus it is probably faster in finding the global minimum. We also use this acquisition function for BO in our simulation and the experiment.\\

\section{Closed-loop model selection} \label{sec:MSPE}
Our goal is to optimize the model's kernel and its hyperparameters with respect to the cost functional $C(\y_{0:k},\u_{0:k})$. Thus, in contrast to the classical kernel selection problem, where the kernel is selected to minimize the state prediction error, our goal here is not to get the most accurate model but the one that achieves the best closed-loop behavior. We now describe the proposed overall procedure for the kernel selection to optimize the closed-loop behavior; we then describe each step in detail. 

We start with an initial kernel $\mathfrak{K}$ with hyperparameters~$\bm{\varphi}$, and obtain the control law for the system~\cref{for:sys} using~\cref{for:ctrllaw} with the model output $\m_k=\M(\u_{0:k-1},\y_{0:k},\mathfrak{K},\bm{\varphi})$. This control law is then applied to the actual system, and the cost function~\cref{for:cfcn} is evaluated after performing the control task. 
Depending on the obtained cost value, BO suggests a new kernel and corresponding hyperparameters for the kernel-based model $\M$ in order to minimize the cost function on the actual system. 
With this model, the control task is repeated and, based on the cost evaluation, BO suggests the next kernel and hyperparameters. This procedure is continued until a maximum number of task evaluations is reached or the user rates the closed-loop performance as sufficient enough. 
We now describe the above three steps, i.e. initialization, evaluation and optimization, in detail.

\subsection{Initialization}
We define a set $\K=\{\mathfrak{K}^1,\ldots,\mathfrak{K}^{n_\mathfrak{K}}\}$ of $n_\mathfrak{K}\in\N$ kernel candidates $\mathfrak{K}^j$ that we want to choose the kernel from for our kernel-based model. BO will be used to select the kernel with the best closed-loop performance in this set.
\begin{rem}
The selection of possible kernels can be based on prior knowledge about the system, e.g. smoothness with the Mat\'ern kernel or number of equlibria using a polynomial kernel, see~\cite{beckers:cdc2016} and \cite{bishop2006pattern} for general properties, respectively.
\end{rem}
In addition, each kernel depends on a set of hyperparameters. Since the number of hyperparameters could be different for each kernel, we define a set of sets $\mathcal{P}=\{\Phi^1,\ldots,\Phi^{n_\mathfrak{K}}\}$ such that $\Phi^j\subset\R^{n_{\Phi^j}}$ is a closed set. Here, ${n_{\Phi^j}}$ represents the number of hyperparameters for the kernel $\mathfrak{K}^j$.
Moreover, we assume that $\Phi^j$ is a valid \textit{hyperparameter set}.
\begin{defn}
The set $\Phi$ is called a  hyperparameter set for a kernel function $\mathfrak{K}$ iff the set $\Phi$ is a domain for the hyperparameters of $\mathfrak{K}$.
\end{defn}
For the first evaluation of the closed-loop, the kernel-based model function $\M$ is created with an initial kernel $\mathfrak{K}^j$ of the set $\K$ and hyperparameters $\bm{\varphi}^j\in\Phi^j$ with $j\in\{1,\ldots,n_\mathfrak{K}\}$.
\begin{rem}
One potential way to select the initial kernel and hyperparameters is to set them equal to the kernel and hyperparameters of a prediction model that is optimized with respect to a loss function, e.g., using cross-validation or maximization of the likelihood function~\cite{bishop2006pattern}.
\end{rem}
\subsection{Task Evaluation}
For the $i$-th task evaluation, BO determines an index value~$j\in\{1,\ldots,n_\mathfrak{K}\}$ and a $\bm{\varphi}^j\in\Phi^j$. The control law~\cref{for:ctrllaw} for the kernel-based model $\M$, with the determined kernel~$\mathfrak{K}^j$ and hyperparameters $\bm{\varphi}^j$, is applied to the system~\cref{for:sys}
\begin{align}
    \x_{k+1}&=\f(\x_k,\bm{h}(\y_k,\M(\u_{0:k-1},\y_{0:k},\mathfrak{K}^j,\bm{\varphi}^j))\notag\\
    \y_k&=\g(\x_k,\u_k)\text{ for }k=\{0,\ldots,n-1\}\notag
\end{align}
with fixed $\x_0\in\R^{n_x}$. 
\begin{rem}
We focus here on a single, fixed initial state $x_0$. However, multiple (close by) initial states can be considered by using the expected cost across all initial states.
\end{rem}
The corresponding input and output sequences $\u_{0:k}$ and~$\y_{0:k}$, respectively, are recorded. Afterwards, the cost function given by $C(\y_{0:k},\u_{0:k})$ is evaluated.
\subsection{Model Optimization}
In the next step, we use BO to minimize the cost function with respect to the kernel and its hyperparameters, i.e.
\begin{align}
    [\mathfrak{K}^j,\bm{\varphi}^j]=\argmin_{j\in\{1,\ldots,n_\mathfrak{K}\},\bm{\varphi}^j\in\Phi^j}C(\y_{0:k},\u_{0:k}).\label{for:boopt}
\end{align}
Thus, this problem involves continuous and discrete variables in the optimization task whereas classical BO assumes continuous variables only. To overcome this restriction, a modified version of BO is used where the kernel function is transformed in a way such that integer-valued inputs are properly included~\cite{garrido2017dealing}.\\
Based on previous evaluations of the cost function, BO updates the prior and minimizes the acquisition function. The result is a kernel $\mathfrak{K}^j$ and hyperparameters $\bm{\varphi}^j$ which are used in the model function $\M(\u_{0:k-1},\y_{0:k},\mathfrak{K}^j,\bm{\varphi}^j)$. Then, the corresponding control law is evaluated again on the system and the procedure is repeated until a maximum number of task evaluations has been reached or a sufficient performance level has been achieved.
\subsection{Theoretical Analysis}
In this section, we show that, under some additional assumptions, the stability of the closed-loop is preserved during the task evaluation process and that BO converges to the minimum of the closed-loop cost function. Here, we focus on stationary kernels
\begin{align}
    k(\x,\x^\prime)=k\big((\x-\x^\prime)^\top\Sigma^{-1}(\x-\x^\prime)\big),\,\x,\x^\prime\in\R^{n_x}
\end{align}
with lengthscales $\bm{\varphi}\in\R^{n_\varphi}_{>0}$ and $\Sigma=\diag(\varphi_1,\ldots,\varphi_{n_x})$. Stationary kernels can always be expressed as a function of the difference between their inputs and they are a common choice for kernel-based models~\cite{bishop2006pattern}. 
\begin{assum}
    Let $\Verts{\f}_{\mathfrak{K}^*_{\varphi^*}}< \infty$ and the selected control law~\cref{for:ctrllaw}, based on the model $\M$ with stationary kernel $\mathfrak{K}^*$ and hyperparameters $\bm{\varphi}^*\in\R^{n_\varphi}_{>0}$, guarantees that $\Verts{\y_k}\leq r_y\in\R_{>0}$ for the given system~\cref{for:sys} for $k>n_1\in\N$. \label{ass:stab}
\end{assum}
The first part of the assumption, i.e. the bounded reproducing kernel Hilbert space (RKHS) norm, is a measure of smoothness of the function with respect to the kernel $\mathfrak{K}$ with hyperparameters $\bm{\varphi}^*\in\R^{n_\varphi}_{>0}$. It is a common assumption for stabilizing controllers using kernel-based methods and is discussed in more detail in~\cite{berkenkamp2016safe}.
Controllers that satisfy this property for nonlinear, unknown systems are given, e.g. by~\cite{chowdhary2012model,berkenkamp2016safe,beckers2019automatica}. The focus on stationary kernels is barely restrictive as many successful applied kernels for nonlinear control are stationary. 
\begin{lem}
    With~\cref{ass:stab}, there exists a non-empty set~$\K$ and a hyperparameter set $\Phi^1\supset\{\bm{\varphi}^*\}$ such that $\forall\mathfrak{K}^j\in\K$, for all $\bm{\varphi}^j\in\Phi^j$ the boundedness $\Verts{y_k}\leq r_y$ of the system~\cref{for:sys} for $k>n_1$ is preserved.\label{lem:pres}
\end{lem} 
This lemma guarantees that there exists a kernel set $\K$ and a set $\mathcal{P}$ of hyperparameters that contains the stabilizing kernel~$\mathfrak{K}^*$ and the hyperparameter $\bm{\varphi}^*$ of~\cref{ass:stab}.
Thus, the proposed method can be applied to existing kernel-based control methods without losing achieved guarantees. Before we start with the proof, the following lemma is recalled.
\begin{lem}[{{\cite[Lemma 4]{bull2011convergence}}}]\label{lem:conv}
If $\f\in\mathcal{H}_{\mathfrak{K}_\varphi}$ then $\f\in\mathcal{H}_{\mathfrak{K}_{\varphi^\prime}}$ holds for all $0<\varphi_i^\prime\leq\varphi_i,\forall i\in\{1,\ldots,n_\varphi\}$, and
\begin{align}
    \Verts{\f}^2_{\mathfrak{K}_{\varphi^\prime}}\leq \left( \prod_{i=1}^{n_\varphi}\frac{\varphi_i}{\varphi^\prime_i}\right)\Verts{\f}^2_{\mathfrak{K}_\varphi}.\label{for:lemconv}
\end{align}
\end{lem}
\begin{proof}[{{\Cref{lem:pres}}}]
\Cref{ass:stab} inherently guarantees that at least one kernel $\mathfrak{K}^1=\mathfrak{K}^*$ exist that preserves the boundedness of the system such that we define~$\K=\{\mathfrak{K}^1\}$. Since~\cref{ass:stab} ensures that $\Verts{\f}_{\mathfrak{K}^*_{\varphi^*}}$ is bounded and with~\cref{lem:conv}, the mapping $\f\in\mathcal{H}_{\mathfrak{K}^1_{\varphi^*}}$ and, thus $\Verts{\f}_{\mathfrak{K}^1_{\underline{\varphi}}}$ is bounded for $\underline{\varphi}_i\in\R_{>0},\forall i$ where $\underline{\varphi}_i<\varphi_i^*,\forall i$. For an upper bound, there exist $\overline{\varphi}_i\in\R_{>0},\forall i$ such that $\varphi_i^*<\overline{\varphi}_i$ and $\f\in\mathcal{H}_{\mathfrak{K}^1_{\overline{\varphi}}}$, following~\cref{lem:conv}. Thus, we define the set 
\begin{align}
   \Phi^1=\{\bm{\varphi}\colon\underline{\varphi}_i\leq \varphi_i \leq \overline{\varphi}_i,\forall i\}\label{for:supset}
\end{align}
as proper superset of $\bm{\varphi^*}$. Based on this set, $\Verts{\f}_{\mathfrak{K}^1_{\varphi}}<\infty$ for all $\bm{\varphi}\in\Phi^1$ that guarantees the boundedness.
\end{proof}
Consequently, with~\cref{ass:stab}, the stability of the control loop is preserved during the task evaluation. Furthermore, the minimum cost is not worse than the initial cost after BO as stated in the following.
\begin{cor}
    Let $C_{cl}$ be the minimum cost~\cref{for:cfcn} after BO~\cref{for:boopt} with $\K=\{\mathfrak{K}^1=\mathfrak{K}^*\}$ and $\Phi^1$ of~\cref{for:supset}. Let $C_{ol}$ be the initial cost based on the control with kernel $\mathfrak{K}^*$ and hyperparameter~$\bm{\varphi}^*$ then $C_{cl}\leq C_{ol}$ holds.
\end{cor}
\begin{proof}
    Since $C_{cl}$ is the minimum cost after BO that starts with the initial, data-based selected kernel $\mathfrak{K}^*$ and hyperparameter $\bm{\varphi}^*$, it clearly follows that $C_{cl}\leq C_{ol}$ because of $\mathfrak{K}^*\in\K$ and $\bm{\varphi}^*\in\Phi^1$.
\end{proof}

We now show that BO can converge to the global minimum of the cost function $C$ under specific conditions starting with the following assumption.
\begin{assum}
    The RKHS norm of the cost function is bounded, i.e. $\Verts{C}_{\mathfrak{K}}\!\leq r\in\!\R_{>0}$ with respect to the kernel~$\mathfrak{K}$ of the GP~\cref{for:likopt} that is used as prior $C\sim\mathcal{GP}(0,\mathfrak{K})$ of the Bayesian optimization~\cref{for:boopt}.\label{ass:boundedC}
\end{assum}
Intuitively, \Cref{ass:boundedC} states that the kernel of the GP for BO is selected such that the GP can properly approximate the cost function. This sounds paradoxical since the cost function is unknown because of the unknown system behavior. However, there exist some kernels, so called universal kernels, which can approximate at least any continuous function arbitrarily precisely~\cite[Lemma 4.55]{steinwart2008support}.
\begin{lem}[{{~\cite{srinivas2012information}}}]\label{lem:bayopt}
With~\cref{ass:boundedC}, BO in~\cref{for:boopt} with upper confidence bound acquisition function~\cite[Eq.(6)]{srinivas2012information} converge with high probability to the global minimum of $C$.
\end{lem}
\section{Evaluation}
In this section, we present a simple illustrative example that highlights our closed-loop model selection approach for kernel-based models. In addition, an example with a 3-DOF robot demonstrates the applicability of the proposed approach to hardware testbeds. BO is used with the expected-improvement-plus as acquisition function because of its satisfactory performance in practical applications, see~\cite{bull2011convergence}, using a GP as prior.
\subsection{Simulation}
\label{sec:sim}
Consider the following one-dimensional system
\begin{align}
\begin{split}
x_{k+1}&=\exp(-\frac{1}{100}x_k^2)\sin(x_k)+\frac{1}{3}x_k+u_k,\, y_k=x_k\label{for:1dimsys}
\end{split}
\end{align}
with state $x_k$ and control $u_k$ at time $k$. For the purpose of this example, we assume that the system dynamics in \eqref{for:1dimsys} are unknown yet we wish to avoid a high-gain control approach due to its unfavorable properties~\cite{isidori2013nonlinear}, and use the proposed closed-loop model selection framework to optimize the control performance. As control law, a feedback linearization 
\begin{align}
u_k=-\hat{f}(x_k\vert\M,\D)+\frac{1}{2}x_k\label{for:1dimctrl}
\end{align}
is applied with the prediction $\hat f$ of a Support vector machine model $\M$. The data set $D$ consists of 11 homogeneously distributed training pairs $\{x^j_k,x^j_{k+1}\}_{j=1}^{11}$ of the system~\cref{for:1dimsys} in the interval $x_k\in[-10,10]$ with $u_k=0$. The linear, polynomial (cubic) and the Gaussian kernel are selected as possible kernel candidates, see~\Cref{tab:kernel_cand} for details.
The Gaussian kernel possesses one hyperparameter $\varphi_1$ which is a scaling factor for the data. In addition, the regression of the SVM depends on a hyperparameter $\varphi_2$ that defines the smoothness of the prediction and affects the number of support vectors, see~\cite{kecman2001learning}.\\
First, we evaluate a classical, data-based procedure which optimizes the kernel and the hyperparameters with respect to the cross-validation loss function~\cite{steinwart2008support} based on the training data only. Using BO, a minimum loss of $0.9127$ is found using the linear kernel with $\varphi_2=0.0336$,~\Cref{tab:comp}. 
  \begin{table}[b]
  \vspace{-0.5cm}
	\caption{Kernel candidates\label{tab:kernel_cand}\vspace{-0.2cm}}
	\begin{tabularx}{\columnwidth}{p{2.5cm}p{5cm}}
		\toprule
		Kernel & Formula\\
		\midrule
		Linear & $\mathfrak{K}(\x,\x^\prime)=\x^\top\x^\prime$\\
		Polynomial (cubic) & $\mathfrak{K}(\x,\x^\prime)=(1+\x^\top\x^\prime)^3$\\
		Gaussian & $\mathfrak{K}(\x,\x^\prime)=\exp(-\frac{\Verts{\x-\x^\prime}^2}{\varphi_1^2}),\,\varphi_1\in\R$\\
	    \bottomrule
	\end{tabularx}
	\vspace{-0.2cm}
\end{table}
  \begin{table}[b]
  \vspace{-0.0cm}
	\caption{Comparison between data-based, data-based with additional training data and closed-loop optimization\label{tab:comp}\vspace{-0.2cm}}
	\begin{tabularx}{\columnwidth}{lllll}
		\toprule
		Method		&	Selected kernel & $\varphi_1,\varphi_2$ & Loss & Cost\\
		\midrule
		Data-based & Linear & $-,0.034$ & $0.913$ & $204.477$\\
		Data-based AT & Linear & $-,0.301$ & $0.09$ & $199.634$\\
		Closed-loop & Gaussian & $\overline{2.333},\overline{0.013}$ & $\overline{2.491}$ & $\overline{16.410}$\\
	    \bottomrule
	\end{tabularx}
	\vspace{-0.2cm}
\end{table}	
Using this linear model in the control loop with the nonlinear system~\cref{for:1dimsys} and control law~\cref{for:1dimctrl} for $x_0=3$, the control error remains above zero, see~\cref{fig:sim1}. With the cost function $C=\sum_{k=0}^{9}k x_k^2$, the trajectory generates a cost of $204.4769$.\\
In comparison, the hyperparameters and the kernel are optimized with the proposed method. For this purpose, we evaluate the performance of the closed-loop system and use BO to compute the next promising kernel and hyperparameter combination.~\Cref{fig:sim2} shows the mean and standard deviation of 20 repetitions over 50 trials each. The repetitions are run since BO exploration of the cost is also affected by randomness. The cost is reduced to a mean value of~$C=\overline{16.410}$ and the loss is $\overline{2.491}$.~\Cref{fig:sim1} shows that the regression is more accurate which results in a reduced control error.~\Cref{tab:comp} also presents the results for adding the collected data of all the 50 trials to the existing data to redefine the model (Data-based AT). Even with more training data, the data-based optimization favors the linear kernel.
\subsubsection{Discussion}
The example demonstrates that the optimization based on the training data only can lead to a reduced performance of the closed-loop system.~\Cref{tab:comp} clearly shows that the data-based optimization results in a smaller loss with the linear kernel but generates a higher cost of the closed-loop system. In comparison, the closed-loop optimization finds a set of hyperparameters with the Gaussian kernel that significantly reduced the control error even if the loss of the model is higher. Thus, especially in the case of sparse data, the data-based optimization can misinterpret the data which can be avoided with the closed-loop model selection. We observe that at the beginning of the closed-loop optimization, BO switches a lot between the kernels and towards the
end, it focus on the hyperparameters. Using the data which is obtained during the 50 trials to refine the model in data-based manner only slightly improves the performance but heavily increases the computational time of the kernel-based model due to the larger training data set. 
\begin{figure}[t]
\begin{center}
\begin{tikzpicture}
\pgfplotsset{yticklabel style={text width=2em,align=right}}
\begin{axis}[
  xlabel={Time step},
  ylabel={Control error},
  legend pos=north west,
  grid style={dashed,gray},
  grid = both,
       width=\columnwidth,
  height=4cm,
  ymin=-0.2,
  ymax=4,
  xmin=-0.2,
  xmax=9.4,
  legend style={font=\footnotesize},
  legend cell align={left},
  legend pos=north east]
\addplot+[name path=varp1, color=blue,opacity=0.3, no marks] table [x index=0,y expr=\thisrowno{2}+5*\thisrowno{3}]{data/figure1_dyn.dat};
\addplot+[name path=varm1, color=blue,opacity=0.3, no marks] table [x index=0,y expr=\thisrowno{2}-5*\thisrowno{3}]{data/figure1_dyn.dat};
\addplot[blue,opacity=0.3] fill between[ of = varm1 and varp1]; 
\addplot[color=blue,line width=1pt,mark=o,mark size=3] table [x index=0,y index=2]{data/figure1_dyn.dat};
\addplot[color=red,line width=1pt,mark=+,mark size=5] table [x index=0,y index=1]{data/figure1_dyn.dat};
\addplot[color=red,line width=1pt] coordinates {(20,20)};
\addplot[color=blue,line width=1pt] coordinates {(20,20)};
\legend{,,,,,after data-based optimization,after closed-loop optimization};
\end{axis}
\end{tikzpicture}
	\vspace{-0.7cm}\caption{Control error (top) and system model (bottom) using closed-loop model optimization for 20 repetitions with mean and $5\sigma$ deviation (blue) and data-based model selection (red).}\vspace{-0.6cm}
	\label{fig:sim1}
\end{center}
\end{figure}
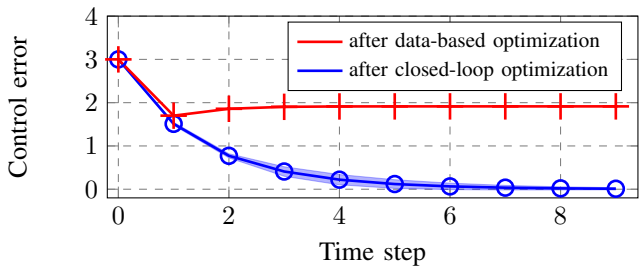
\begin{figure}[t]
\begin{center}
	\begin{tikzpicture}
\begin{axis}[
  xlabel={Task evaluations},
  ylabel={Min. cost $C$},
  legend pos=north west,
  grid style={dashed,gray},
  grid = both,
       width=\columnwidth,
  height=4cm,
  ymin=0,
  ymax=400,
  xmin=1,
  xmax=50,
  legend style={font=\footnotesize},
  legend cell align={left}]
\addplot+[name path=varp1, color=blue,opacity=0.3, no marks] table [x index=0,y expr=\thisrowno{1}+\thisrowno{2}]{data/figure1_bo.dat};
\addplot+[name path=varm1, color=blue,opacity=0.3, no marks] table [x index=0,y expr=\thisrowno{1}-\thisrowno{2}]{data/figure1_bo.dat};
\addplot[color=blue,line width=1pt] table [x index=0,y expr=\thisrowno{1}]{data/figure1_bo.dat};
\addplot[blue,opacity=0.3] fill between[ of = varm1 and varp1]; 
\end{axis}
\end{tikzpicture} 
	\vspace{-0.6cm}\caption{Minimum of the cost function over the number of trials for 20 repetitions for the closed-loop model selection algorithm.}\vspace{-0.5cm}
	\label{fig:sim2}
\end{center}
\vspace{-0.4cm}
\end{figure}
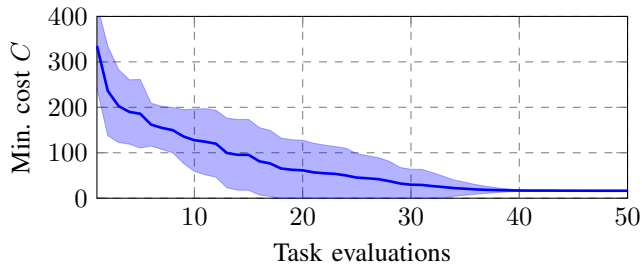
\subsection{Experiment}
\label{sec:exp}
\subsubsection{Setup}
For the experimental evaluation, we use the \mbox{3-dof} SCARA robot CARBO as pictured in~\cref{fig:figure_roboter}. The links between the joints have a length of $\SI{0.3}{\meter}$ and a spoon is attached at the end effector of the robot.
The goal is to follow a given trajectory as precisely as possible without using high feedback gains, which might result in several practical disadvantages, see~\cite{nguyen2008computed}. Therefore,  a precise model of the system's dynamics is necessary. Since the modeling of the nonlinear fluid dynamics with a parametric model would be very time consuming, we use a computed torque control method based on a GP model which allows high performance tracking control while also being able to guarantee the stability of the control loop~\cite{beckers2019automatica}. Underlying, a low level PD-controller enforces the generated torque by regulating the voltage based on a measurement of the current. The controller is implemented in MATLAB/Simulink on a Linux real-time system with a sample rate of $\SI{1}{\milli\second}$. For the implementation of the GP model, we use the GPML toolbox\footnote{http://www.gaussianprocess.org/gpml/code}. The desired trajectory follows a circular stirring movement through the fluid with a frequency of $\SI{0.5}{hz}$.\\
\textbf{Modeling:} Here, we use a Gaussian process model $\M$ as kernel-based model technique based on 223 collected training points. The data is collected around the desired trajectory using a high gain controller. The placement of the training points heavily influences the control performance. However, the proposed approach focuses on improving the performance based on existing data. Each data pair consists of the position and velocity of all joints $[\q,\dq]^\top$ and the corresponding torque for the $i$-th joint, $\tau_i$. Since the GP produces one-dimensional outputs only, 3 GPs are used in total for the modeling of the robot's dynamics. Each GP $i=1,\ldots,3$ uses a squared exponential kernel 
\begin{align}
\mathfrak{K}(\x,\x^\prime)={\varphi}_i^2\exp\left(\frac{-\Verts{\x-\x^\prime}^2}{\varphi_{i+3}^2}\right),\,\varphi_i\in\R\setminus\{0\}
\end{align}
that can approximate any continuous function arbitrarily exactly. With $\bm{\varphi}=[\varphi_1,\ldots,\varphi_6]$ and the signal noise $\bm{\sigma}_n\in\R^3$, see~\cite{rasmussen2006gaussian}, a total number of 9 parameters must be optimized. In contrast to the simulation, the kernel is fixed to reduce the optimization space and thus, the number of task evaluations.
\begin{figure}[t]
\vspace{0.2cm}
\begin{center}
	\includegraphics[width=0.65\columnwidth]{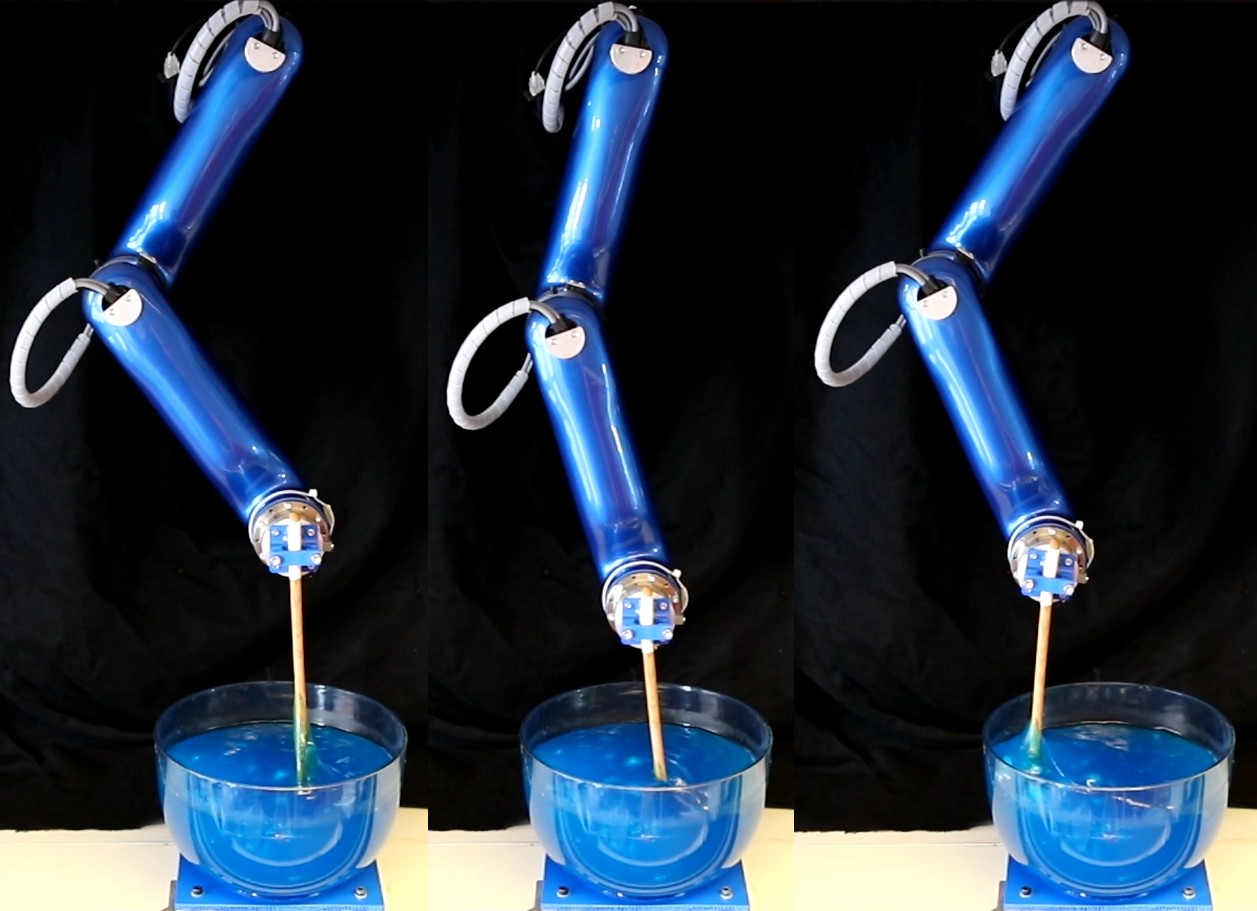}
	\vspace{-0.2cm}\caption{Stirring with the 3-dof SCARA robot CARBO.}\vspace{-0.4cm}
	\label{fig:figure_roboter}
\end{center}
\vspace{-0.4cm}
\end{figure}
\textbf{Control law:} The control input, i.e. the torque $\bm{\tau}(\dq,\q)$ for all joints, is generated based on an estimated parametric model and the mean prediction $\bm{\mu}$ of the GP model as feed-forward component and a low gain PD-feedback part
\begin{align}
\bm{\tau}_d=\hat H\ddq_d+\hat C\dq_d+\hat g+\bm{\mu}(\dq,\q\vert\M)-K_d\de-K_p\e.
\end{align}
Here, the desired trajectory is given by $\q_d,\dq_d$ and $\ddq_d$ with the error $\de=\dq_d-\dq,\e=\q_d-\q$. The feedback matrices are given by $K_p=\diag([60,40,10])$ and $K_d=\diag([1,1,0.4])$. The estimated parametric model is derived from a mathematical model where the parameters are physically measured. For the discretization of the control input, a zero-order method is used. For more details see~\cite{beckers2019automatica}. 
\subsubsection{Evaluation}
The evaluation of the performance of the closed-loop is based on the cost function
\begin{align}
C=\frac{1}{2000}\sum_{k=0}^{2000}\e(kT)^\top\e(kT)
\end{align}
with $T=\SI{1}{\milli\second}$. Therefore, the cost function is a measure for the tracking accuracy of the stirring movement. We consider as kernel candidate the squared exponential kernel, such that only the hyperparameters $\bm{\sigma}_n,\bm{\varphi}$ are optimized.~\Cref{tab:comp_exp} shows the
  \begin{table}[b]
	\vspace{-0.4cm}\caption{Comparison: data-based and closed-loop optimization\label{tab:comp_exp}\vspace{-0.2cm}}
	\begin{tabularx}{\columnwidth}{lll}
		\toprule
		Value		&	Data-based & Closed-loop\\
		\midrule
		$\bm{\sigma}_n$ & $[0.10,\num{3e-3},\num{6e-4}]$ & $[0.20,\num{4e-3},\num{3e-4}]$\\
		$\varphi_{1,2,3}$ & $[3.49,1.42,2.87]$ & $[2.61,1.68,5.70]$\\
		$\varphi_{4,5,6}$ & $[1.21,0.25,0.27]$ & $[ 0.80,0.27,0.29]$\\
		Log. likelihood & $[89,-121,-176]$ & $[115,-113,-136]$ \\
		Cost (Tracking error) & $1.49$ & $\bm{1.05}$\\
	    \bottomrule
	\end{tabularx}
	\vspace{-0.2cm}
\end{table}	
comparison between the data-based and the closed-loop optimization. In the data-based case, the hyperparameters are optimized based on a gradient method to minimize the log likelihood function (in this case, BO of the hyperparameters results in the same values). In contrast, BO is used to minimize the tracking error in the closed-loop optimization. The initial values of the hyperparameters are set to the values of the data-based optimization. The bounds are defined as the $0.5$ and $2$ times of the initial values. The evolution of the minimum cost over the trials, where each trial is a single stirring movement, is shown in~\cref{fig:figure_exp_trial}.\newline The comparison of the joint position error for the data-based and closed-loop optimization is shown in~\cref{fig:figure_exp_error}. 
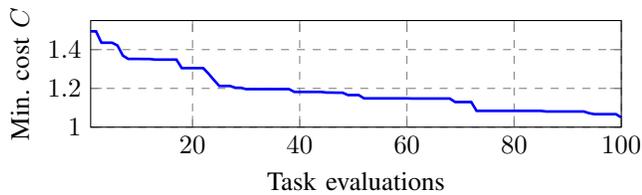
\begin{figure}
\begin{center}
\vspace{0.2cm}
	\begin{tikzpicture}
\begin{axis}[
  xlabel={Task evaluations},
  ylabel={Min. cost $C$},
  legend pos=north west,
  grid style={dashed,gray},
  grid = both,
       width=\columnwidth,
  height=3cm,
  ymin=1,
  ymax=1.55,
  xmin=1,
  xmax=100,
  legend style={font=\footnotesize},
  legend cell align={left}]
\addplot[color=blue,line width=1pt,no marks] table [x index=0,y index=1]{data/figure_exp_trials.dat};
\end{axis}
\end{tikzpicture} 
	\vspace{-0.8cm}\caption{Minimum of the cost function over the number of trials.}\vspace{-0.4cm}
	\label{fig:figure_exp_trial}
\end{center}
\vspace{-0.4cm}
\end{figure}
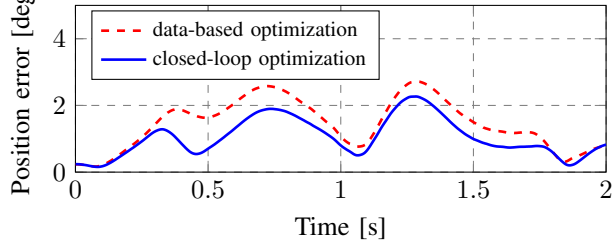
\begin{figure}
\begin{center}
	\begin{tikzpicture}
\begin{axis}[
  xlabel={Time [s]},
  ylabel={Position error [deg]},
  legend pos=north west,
  grid style={dashed,gray},
  grid = both,
       width=\columnwidth,
  height=3.8cm,
  ymin=0,
  ymax=5,
  xmin=0,
  xmax=2,
  legend style={font=\footnotesize},
  legend cell align={left}]
\addplot[color=red,dashed,line width=1pt,no marks] table [x index=0,y index=1]{data/figure_exp_error.dat};
\addplot[color=blue,line width=1pt,no marks] table [x index=0,y index=2]{data/figure_exp_error.dat};
\legend{data-based optimization,closed-loop optimization};
\end{axis}
\end{tikzpicture} 
	\vspace{-0.4cm}\caption{Comparison of the root square position error of all joints.}\vspace{-0.4cm}
	\label{fig:figure_exp_error}
\end{center}
\vspace{-0.4cm}
\end{figure}
\subsubsection{Discussion}
After 100 trials, the tracking error is decreased by 30\% through the optimization of the Gaussian process model only. Even if the resulting hyperparameters are sub-optimal with respect to the likelihood function, see~\Cref{tab:comp_exp}, the performance of the closed-loop is significantly improved. In comparison to collecting more training data to improve the model, the proposed method does not increase the computational burden of the Gaussian process prediction which is often critical in real-time applications. Since only the model is adapted, the properties of the closed-loop control architecture are also preserved. 

\section*{Conclusion}
In this paper, we present a framework for the model selection for kernel-based models to directly optimize the overall closed-loop control performance. For this purpose, the kernel and its hyperparameters are optimized using Bayesian optimization with respect to a cost function that evaluates the performance of the closed-loop. It is shown that this approach allows to preserve the control architecture properties as only the model is adapted. Simulations and hardware experiments demonstrate the advantages of the proposed approach to data-based model selection techniques. 

\section*{Acknowledgments}
The research has received funding from the ERC Starting Grant ``Con-humo'' n\textsuperscript{o}337654 and BaCaTec grant 9-[2018/1].

\bibliography{mybib}

\begin{thebibliography}{10}

\bibitem{bishop2006pattern}
C.~M. Bishop {\em et~al.}, {\em Pattern recognition and machine learning},
  vol.~4.
\newblock Springer New York, 2006.

\bibitem{abbeel2006using}
P.~Abbeel, M.~Quigley, and A.~Y. Ng, ``Using inaccurate models in reinforcement
  learning,'' in {\em International Conference on Machine Learning}, 2006.

\bibitem{geversaa2005identification}
M.~Gevers, ``Identification for control: From the early achievements to the
  revival of experiment design,'' {\em European journal of control}, vol.~11,
  pp.~1--18, 2005.

\bibitem{hjalmarsson1996model}
H.~Hjalmarsson, M.~Gevers, and F.~De~Bruyne, ``For model-based control design,
  closed-loop identification gives better performance,'' {\em Automatica},
  vol.~32, no.~12, pp.~1659--1673, 1996.

\bibitem{ljung1998system}
L.~Ljung, ``System identification,'' in {\em Signal analysis and prediction},
  pp.~163--173, Springer, 1998.

\bibitem{chowdhary2012model}
G.~Chowdhary, J.~How, and H.~Kingravi, ``Model reference adaptive control using
  nonparametric adaptive elements,'' in {\em Conference on Guidance Navigation
  and Control}, 2012.

\bibitem{umlauft:cdc2017}
J.~Umlauft, T.~Beckers, M.~Kimmel, and S.~Hirche, ``Feedback linearization
  using {G}aussian processes,'' in {\em IEEE Conference on Decision and
  Control}, 2017.

\bibitem{beckers:ifacwc2017}
T.~Beckers, J.~Umlauft, and S.~Hirche, ``Stable model-based control with
  {Gaussian} process regression for robot manipulators,'' in {\em IFAC World
  Congress}, 2017.

\bibitem{narendra1990identification}
K.~S. Narendra and K.~Parthasarathy, ``Identification and control of dynamical
  systems using neural networks,'' {\em IEEE Transactions on neural networks},
  vol.~1, no.~1, pp.~4--27, 1990.

\bibitem{bansal2016learning}
S.~Bansal, A.~K. Akametalu, F.~J. Jiang, F.~Laine, and C.~J. Tomlin, ``Learning
  quadrotor dynamics using neural network for flight control,'' in {\em IEEE
  Conference on Decision and Control}, 2016.

\bibitem{aastrom2013adaptive}
K.~J. {\AA}str{\"o}m and B.~Wittenmark, {\em Adaptive control}.
\newblock Courier Corporation, 2013.

\bibitem{clarke1985generalized}
D.~Clarke, P.~Kanjilal, and C.~Mohtadi, ``A generalized {LQG} approach to
  self-tuning control part i. aspects of design,'' {\em International Journal
  of Control}, vol.~41, no.~6, pp.~1509--1523, 1985.

\bibitem{bristow2006survey}
D.~A. Bristow, M.~Tharayil, and A.~G. Alleyne, ``A survey of iterative learning
  control,'' {\em IEEE control systems magazine}, vol.~26, no.~3, pp.~96--114,
  2006.

\bibitem{lewis2013reinforcement}
F.~L. Lewis and D.~Liu, {\em Reinforcement learning and approximate dynamic
  programming for feedback control}, vol.~17.
\newblock John Wiley \& Sons, 2013.

\bibitem{Calandra2015a}
R.~Calandra, A.~Seyfarth, J.~Peters, and M.~P. Deisenroth, ``{Bayesian}
  optimization for learning gaits under uncertainty,'' {\em Annals of
  Mathematics and Artificial Intelligence}, vol.~76, no.~1, pp.~5--23, 2015.

\bibitem{recht2018tour}
B.~Recht, ``A tour of reinforcement learning: The view from continuous
  control,'' {\em Annual Review of Control, Robotics, and Autonomous Systems},
  2018.

\bibitem{bansal2017goal}
S.~Bansal, R.~Calandra, T.~Xiao, S.~Levine, and C.~J. Tomlin, ``Goal-driven
  dynamics learning via {Bayesian} optimization,'' in {\em IEEE Conference on
  Decision and Control}, 2018.

\bibitem{suykens2001optimal}
J.~A. Suykens, J.~Vandewalle, and B.~De~Moor, ``Optimal control by least
  squares support vector machines,'' {\em Neural networks}, vol.~14, no.~1,
  pp.~23--35, 2001.

\bibitem{berkenkamp2016safe}
F.~Berkenkamp, R.~Moriconi, A.~P. Schoellig, and A.~Krause, ``Safe learning of
  regions of attraction for uncertain, nonlinear systems with {Gaussian}
  processes,'' in {\em IEEE Conference on Decision and Control}, 2016.

\bibitem{rasmussen2006gaussian}
C.~E. Rasmussen and C.~K. Williams, {\em {Gaussian} processes for machine
  learning}, vol.~1.
\newblock MIT press Cambridge, 2006.

\bibitem{shahriari2016taking}
B.~Shahriari, K.~Swersky, Z.~Wang, R.~P. Adams, and N.~De~Freitas, ``Taking the
  human out of the loop: A review of {Bayesian} optimization,'' {\em IEEE},
  vol.~104, no.~1, pp.~148--175, 2016.

\bibitem{mockus2012bayesian}
J.~Mockus, {\em {B}ayesian approach to global optimization: theory and
  applications}, vol.~37.
\newblock Springer Science \& Business Media, 2012.

\bibitem{bull2011convergence}
A.~D. Bull, ``Convergence rates of efficient global optimization algorithms,''
  {\em Journal of Machine Learning Research}, vol.~12, no.~Oct, pp.~2879--2904,
  2011.

\bibitem{beckers:cdc2016}
T.~Beckers and S.~Hirche, ``Equilibrium distributions and stability analysis of
  {G}aussian process state space models,'' in {\em IEEE Conference on Decision
  and Control}, 2016.

\bibitem{garrido2017dealing}
E.~C. Garrido-Merch{\'a}n and D.~Hern{\'a}ndez-Lobato, ``Dealing with
  integer-valued variables in {Bayesian} optimization with {Gaussian}
  processes,'' {\em arXiv preprint arXiv:1706.03673}, 2017.

\bibitem{beckers2019automatica}
T.~Beckers, D.~Kulić, and S.~Hirche, ``Stable {Gaussian} process based
  tracking control of {Euler}-{Lagrange} systems,'' {\em Automatica}, no.~103,
  pp.~390--397, 2019.

\bibitem{steinwart2008support}
I.~Steinwart and A.~Christmann, {\em Support vector machines}.
\newblock Springer Science \& Business Media, 2008.

\bibitem{srinivas2012information}
N.~Srinivas, A.~Krause, S.~M. Kakade, and M.~W. Seeger, ``Information-theoretic
  regret bounds for {Gaussian} process optimization in the bandit setting,''
  {\em IEEE Transactions on Information Theory}, vol.~58, no.~5,
  pp.~3250--3265, 2012.

\bibitem{isidori2013nonlinear}
A.~Isidori, {\em Nonlinear control systems}.
\newblock Springer Science \& Business Media, 2013.

\bibitem{kecman2001learning}
V.~Kecman, {\em Learning and soft computing: support vector machines, neural
  networks, and fuzzy logic models}.
\newblock MIT press, 2001.

\bibitem{nguyen2008computed}
D.~Nguyen-Tuong, M.~Seeger, and J.~Peters, ``Computed torque control with
  nonparametric regression models,'' in {\em American Control Conference},
  2008.

\end{thebibliography}
\bibliographystyle{ieeetr}

\end{document}